\documentclass[a4paper]{article}

\usepackage{amsmath, amssymb}
\usepackage{graphicx}
\usepackage{booktabs}

\newcommand{\ExpectMeas}[2]{\mathbb{E}^{#1}\left[#2\right]}

\newcommand{\keywords}[1]{{\bf Keywords:} {#1}}
\newcommand{\subclass}[1]{{\bf Mathematics Subject Classification (2000):} {#1}}
\newcommand{\JEL}[1]{{\bf JEL:} {#1}}

\newtheorem{thm}{theorem}[section]
\newtheorem{theorem}[thm]{Theorem}
\newtheorem{proposition}[thm]{Proposition}
\newtheorem{corollary}[thm]{Corollary}
\newtheorem{lemma}[thm]{Lemma}
\newenvironment{proof}{{\bf Proof }}{\hfill$\Box$}

\newcommand{\MakeTitle}{\maketitle\newcommand{\and}{$\cdot$ }}

\title{The Impact of the Prior Density on a Minimum Relative Entropy Density: A Case Study with SPX Option Data
	\thanks{
		We would like to thank Olivier Le Courtois, Fran\c{c}ois Quittard-Pinon, Matthias Scherer and Peter Tankov for helpful comments.
        We are particularly grateful to Nabil Kahal\'{e} for his comments and the suggestion to study variance swaps.
	}
}

\author{
	Cassio Neri
	\thanks{Lloyds Banking Group, \texttt{cassio.neri@lloydsbanking.com}.}
	\and
	Lorenz Schneider
	\thanks{Center for Financial Risks Analysis (CEFRA), EMLYON Business School, \texttt{schneider@em-lyon.com}.}
}

\date{\today}

\begin{document}

\MakeTitle

\begin{abstract}
We study the problem of finding probability densities that match given European call option prices.
To allow prior information about such a density to be taken into account, we generalise the algorithm presented in \cite{NeriSchneider2011}
to find the maximum entropy density of an asset price to the relative entropy case.
This is applied to study the impact the choice of prior density has in two market scenarios.
In the first scenario, call option prices are prescribed at only a small number of strikes,
and we see that the choice of prior, or indeed its omission, yields notably different densities.
The second scenario is given by CBOE option price data for S\&P$500$ index options at a large number of strikes.
Prior information is now considered to be given by calibrated Heston, Sch\"obel-Zhu or Variance Gamma models.
We find that the resulting digital option prices are essentially the same as those given by the (non-relative) Buchen-Kelly density itself.
In other words, in a sufficiently liquid market the influence of the prior density seems to vanish almost completely.
Finally, we study variance swaps and derive a simple formula relating the fair variance swap rate to entropy.
Then we show, again, that the prior loses its influence on the fair variance swap rate as the number of strikes increases.

\bigskip


\bigskip

\subclass{91B24 \and 91B28 \and 91B70 \and 94A17}

\bigskip

\JEL{C16 \and C63 \and G13}

\end{abstract}

\section{Introduction}
\label{s:Introduction}

Many financial derivatives are valued by calculating their expected payoff under the risk-neutral measure.
For path-independent derivatives the expectation can be obtained by integrating the product of the payoff function and the density.

If a pricing model has been chosen for a market in which many derivative products are actively quoted,
then often, due to the limited number of model parameters, this model will be unable to perfectly match the market quotes,
and a compromise must be made during model calibration by using some kind of ``best-fit'' criterion.

If no model has been chosen, one can try to imply from the market data for a given maturity a probability density function
that leads, by integration as described above, exactly back to the quoted prices. However, unless the market is perfectly liquid,
there will be infinitely many densities that match the price quotes, and some criterion for the selection of the density will have to be applied.
One such criterion is to choose the density that maximises uncertainty or, in another word, entropy. The idea is that between two densities
matching the constraints imposed by the market prices, the one that is more uncertain -- where ``uncertain'' means, very roughly speaking, spreading probability
over a large interval instead of assigning it to just a few points, where possible -- should be chosen. In general, applying the criterion of entropy
delivers convincing results. For example, on the unit interval $[0,1]$, if no constraints are given,
the density with the greatest entropy -- the maximum entropy density (MED) -- is the uniform density.
On the positive real numbers $[0, \infty[$, if the mean is given as the only constraint, the entropy maximiser is the exponential density.
There is no entropy maximiser over the real numbers $\mathbb{R}$, but if the mean and variance are imposed as constraints,
then the density with largest entropy is a Gaussian normal density.

The concept of entropy has its origins in the works of Boltzmann \cite{Boltzmann1877} in Statistical Mechanics and Shannon \cite{Shannon1948} in Information Theory,
and an important recent application has been by Villani \cite{Villani2009} and others in the field of Optimal Transport.
In Finance, too, entropy has become a popular tool, as a survey of recent literature
(see for example \cite{AvellanedaFriedmanHolmesSamperi1997, BorweinChoksiMarechal2003,
BrodyBuckleyMeister2004, BuchenKelly1996, Frittelli2000, FujiwaraMiyahara2003, Gulko1999a, OrozcoRodriguezSantosa2012}) confirms.

A third approach, which we will use to combine the two approaches described above, is to take a density $p$, which may be near to matching some imposed constraints,
as a prior density and to then find a density $q$ that is ``as close'' as possible to $p$ and exactly matches these constraints.
The criterion we will employ to measure the ``distance'' between the densities $q$ and $p$ is that of relative entropy. Since we are trying to depart as little
as possible from the prior $p$, our goal will be to find the {\it minimum} relative entropy density (MRED) $q$ that matches the given constraints.
Of course, if the prior $p$ already matches the constraints, then we can take as our solution $q = p$, since the relative entropy of $p$ with respect to itself is zero.
Relative entropy was introduced by Kullback and Leibler \cite{KullbackLeibler1951} and is also known as the {\it Kullback-Leibler information number $I = I(q \| p)$}
or {\it $I$-Divergence} \cite{Csiszar1975}. Although it is always non-negative and can be used as measure of distance, it is important to stress that it is not a metric in the
mathematical sense, since usually $I(q \| p) \neq I(p \| q)$, and the triangle inequality is not satisfied.

In the study we carry out in this paper, the prior density function of the asset price for a fixed maturity is given by a model,
such as the Black-Scholes model or the Heston stochastic volatility model. Depending on the model in question, this prior density will be either
directly available in analytical form (in the case of the Black-Scholes model a log-normal distribution), or have to be obtained numerically
(in case of the Heston model via Fourier inversion). The main impact of this will be on computation time, but otherwise the difference is of
minor consequence. The algorithms we propose to calculate the MRED $q$ (with respect to $p$) satisfying some constraints given by European option prices
are extensions of the two algorithms presented in \cite{NeriSchneider2009} and \cite{NeriSchneider2011}.

In the first case, the option data consists of call and digital call prices (section \ref{s:MRED_CallDigital}), and in the second case only of call prices (section \ref{s:MRED_Call}).
If only the call prices are imposed, say $n$ of them, the problem consists in finding the minimum of a real-valued, convex function (the relative entropy function) in $n$ variables.
If one additionally imposes the $n$ prices of digital options at the same strikes, the problem simplifies to a sequence of $n$ one-dimensional root-finding problems.
The multi-dimensional algorithm makes use of the single-dimensional one by fixing the set of call prices,
defining a parameter space $\Omega$ for arbitrage-free digital prices, and then finding the unique density in this family with the smallest relative entropy.

The models we take to generate our prior densities are presented in section \ref{s:CharacteristicFunctionModels}, together with a review of the characteristic function
pricing approach and corresponding Fourier transform techniques. In addition to the two models already mentioned above, we also consider the Sch\"obel-Zhu
stochastic volatility model and the Variance Gamma model.

In section \ref{s:NumericalExamples} we study two market scenarios. In the first one, we take a log-normal Black-Scholes and a Heston density as our priors and calculate
the MREDs that match given call prices. We then compare option prices to those obtained with an MED, and also to those obtained with another log-normal density
that matches the constraints, and observe that the price differences can be substantial.
In the second scenario, we take S\&P$500$ call option prices from the CBOE. This market is very liquid, and for the maturity we consider we have quotes for a large set of strikes.
We calibrate a Heston, Sch\"obel-Zhu and Variance Gamma model to this data and use the densities generated by these models as prior densities.
Then, we calculate the three MREDs for these priors, and compare the digital option prices they give to those given by the original models, those given by an MED,
and finally the market prices themselves, which are available in this case. We observe that it makes almost no difference which model
is chosen for the prior, and that all three MREDs essentially agree with the MED.

In section \ref{s:VarianceSwaps} we study variance swaps and the fair swap rate.
Assuming that the underlying asset follows a diffusion process without jumps, it is possible to relate this rate to the price of a log-contract.
A formula linking it to an integral over call and puts prices at varying strikes is also well known (\cite{CarrMadan1998, DemeterfiDermanKamalZou1999, Gatheral2006}.
Here, we establish a simple formula (Corollary \ref{Cor:FairVarianceSwapRate-Entropy}) that relates the fair variance swap rate to entropy.
We then give an explicit formula (see equation \ref{EiPrimitiveFunction}) for the fair variance swap rate in the case of a (non-relative)
MED in terms of the assumed drift rate and the density's parameters.
In the relative entropy case, we calculate the fair rate numerically and show that for MREDs constrained by data at very few strikes the prior density
can have a significant impact on the fair rate.
However, as in the examples given in section \ref{s:NumericalExamples}, the impact of the prior density diminishes quite strongly as data at more strikes are added as constraints.
Finally, section \ref{s:Conclusion} concludes the article.

\section{Relative Entropy and Option Prices}
\label{s:RelativeEntropy}

In this section we review the concept of relative entropy, which can be regarded as a way of measuring the ``distance'' between two given densities.
Our goal is to apply this measure to the following problem:
Given a prior density $p$, coming for example from a model that fits well, but not exactly,
European option prices observed in the market, how can we deform this density in such a way that it exactly matches these prices,
but stays as close as possible to the original density under the criterion of relative entropy?

\subsection{Relative Entropy}
\label{ss:Relative Entropy}

For two probability distributions $Q$ and $P$ the relative entropy of $Q$ with respect to $P$ is defined by
\begin{equation}
\label{DefinitionOfRelativeEntropyForDistributions}
H(Q \| P) = \left\{
\begin{array}{lc}
\displaystyle \int \ln \frac{\partial Q}{\partial P} dQ = \int \frac{\partial Q}{\partial P} \ln \frac{\partial Q}{\partial P} dP,
\quad & Q \ll P, \\
\\
\displaystyle \infty
& \text{else},
\end{array}
\right.
\end{equation}
where $\partial Q/\partial P$ is the Radon-Nikod\'ym derivative.

From the inequality $S \ln S \geq S - 1$ we have
\begin{equation*}
H(Q \| P)
= \int \frac{\partial Q}{\partial P} \ln \frac{\partial Q}{\partial P} dP
\geq \int \left( \frac{\partial Q}{\partial P} - 1 \right) dP
= \int dQ - \int dP
= 0.
\end{equation*}
We also have $H(Q \| P) = 0$ if, and only if, $Q = P$.
However, relative entropy is not a metric since, in general, $ H(Q \| P) \neq H(P \| Q)$, and the triangle inequality is not satisfied either.
Even the symmetric function $H(Q \| P) + H(P \| Q)$ does not define a metric, since it still does not satisfy the triangle inequality \cite{Csiszar1964}.

The Csisz\'ar-Kullback inequality \cite{ArnoldMarkowichToscaniUnterreiter2000, Csiszar1967} relates
relative entropy to distance between densities in the sense of the $L^1(0,\infty)$ norm:
\begin{equation*}
\| q - p \|_{L^1} := \int_0^\infty | q(S) - p(S) | dS \le \sqrt{ 2 H(q \| p) },
\end{equation*}
where
\begin{equation}
\label{DefinitionOfRelativeEntropyForDensities}
H(q \| p) = \int q(S) \ln \frac{q(S)}{p(S)} dS = \int \frac{q(S)}{pS)} \ln \frac{q(S)}{p(S)} p(S) dS
\end{equation}
is the same definition as \eqref{DefinitionOfRelativeEntropyForDistributions} above in terms of densities,
which means in particular that convergence in the sense of relative entropy implies $L^1$-convergence.

\subsection{Minimizer Matching Option Prices}
\label{ss:GeneralCaseAnyPrior}

We now give a precise formulation of the minimisation problems that we want to solve.
Let $p$ be the prior density on $[0, \infty[$ which is assumed to be strictly positive almost everywhere.

For a fixed underlying asset and maturity $T$, we are given undiscounted prices
$\tilde{C}_1$, ..., $\tilde{C}_n$ of call options at strictly increasing strikes $K_1 < \cdots < K_n$.
For notational convenience, we introduce the ``strikes'' $K_0 := 0$ and
$K_{n+1} := \infty$ and make the convention that $\tilde{C}_0$ is the forward asset price for time $T$ and $\tilde{C}_{n+1} = 0$.

In section \ref{s:MRED_Call} we will determine a density $q$ for the underlying asset price $S(T)$ at maturity which minimises relative entropy $H(q \| p)$ under the constraints
\begin{equation}
\label{ME1}
\ExpectMeas{q}{(S(T) - K_i)^+} = \tilde{C}_i, \quad
\text{{\em i.e.},} \quad
\int_{K_i}^\infty (S - K_i) q(S) dS = \tilde{C}_i \quad
\forall i = 0, ..., n.
\end{equation}

Before that, in section \ref{s:MRED_CallDigital}, we shall assume that undiscounted digital option prices $\tilde{D}_1$, ..., $\tilde{D}_n$ on the same asset, maturity and strikes are also given and we look for $q$ that, in addition, verifies the constraints
\begin{equation}
\label{ME2}
\ExpectMeas{q}{ \mathbf{I}_{\{S(T) > K_i\}} } = \tilde{D}_i, \quad
\text{{\em i.e.},} \quad
\int_{K_i}^\infty q(S) dS = \tilde{D}_i \quad
\forall i = 0, ..., n.
\end{equation}
Again, for ease of notation, we make the convention that $\tilde{D}_0 = 1$ and $\tilde{D}_{n+1} = K_{n+1} \tilde{D}_{n+1} = 0.$

Notice that the constraints \eqref{ME1} and \eqref{ME2} for $i=0$ are consistent with the fact that $q$ is a density (its integral is $1$) and the martingale condition
\begin{equation*}
\ExpectMeas{q}{S(T)} = \int_0^\infty S q(S) dS = \tilde{C}_0.
\end{equation*}

\section{Minimizer Matching Call and Digital Option Prices}
\label{s:MRED_CallDigital}

In this section we review some results stated in \cite{NeriSchneider2009} and provide the base arguments required to prove them in case a prior density $p$ is given and call and digital options prices are prescribed.

In addition, we show how the algorithm presented in \cite{NeriSchneider2009} can be efficiently implemented.
We do not assume that $p$ is given analytically, and therefore the implementation requires numerical integration.
However, the availability of the digital prices allows for an efficient solution locally in each ``bucket'', i.e., interval $[K_i, K_{i+1}[$,
via a one-dimensional Newton-Raphson rootfinder.

Formally applying the Lagrange multipliers theorem, as in \cite{BuchenKelly1996}, it can be ``proven''%
\footnote{Rigourously speaking, the Lagrange multipliers theorem cannot be applied since the relative entropy functional is nowhere continuous \cite{BorweinChoksiMarechal2003}}
that if $q$ minimises relative entropy in respect to $p$, then the Radon-Nikod\'ym derivative $g := \partial Q / \partial P = q / p$ is piecewise exponential.
More precisely, on each interval $[K_i, K_{i+1}[$ the density $q$ is given by
\begin{equation}
\label{MREDDensity}
q(S) = g(S) p(S) = \alpha_i e^{\beta_i S} p(S),
\end{equation}
where $\alpha_i, \beta_i \in {\mathbb R}, \alpha_i > 0$ are parameters that still have to be determined using the following two constraints imposed by the option data, which are an equivalent reformulation of the constraints \eqref{ME1} and \eqref{ME2} given above, but allow for an easy solution.

The first constraint follows directly from \eqref{ME2} and is given by
\begin{equation*}
\alpha_i \int_{K_{i}}^{K_{i+1}} e^{\beta_i S} p(S) dS = \tilde{D}_{i} - \tilde{D}_{i+1}, \quad
\forall i = 0, ..., n.
\end{equation*}
from which we have
\begin{equation}
\label{alpha}
\alpha_i = \frac{\tilde{D}_{i} - \tilde{D}_{i+1}}{\int_{K_{i}}^{K_{i+1}} e^{\beta_i S} p(S) dS} \quad
\forall i = 0, ..., n.
\end{equation}

The second constraint also follows directly from \eqref{ME1} and \eqref{ME2} and is given by
\begin{equation*}
\alpha_i \int_{K_{i}}^{K_{i+1}} S e^{\beta_i S} p(S) dS = (\tilde{C}_{i} + K_{i} \tilde{D}_{i}) - (\tilde{C}_{i+1} + K_{i+1} \tilde{D}_{i+1}), \quad
\forall i = 0, ..., n.
\end{equation*}
Notice that the right hand side of the equation above is the undiscounted price of an ``asset-or-nothing'' derivative that pays the asset price itself if it finishes between the two strikes $K_i$ and $K_{i+1}$ at maturity and zero otherwise.
Substituting $\alpha_i$ from \eqref{alpha} gives
\begin{equation}
\label{beta}
\frac{\int_{K_{i}}^{K_{i+1}} S e^{\beta_i S} p(S) dS}{\int_{K_{i}}^{K_{i+1}} e^{\beta_i x} p(S) dS}
= \frac{(\tilde{C}_{i} + K_{i} \tilde{D}_{i}) - (\tilde{C}_{i+1} + K_{i+1} \tilde{D}_{i+1})}{\tilde{D}_{i} - \tilde{D}_{i+1}},
\end{equation}
which we use as an implicit equation for $\beta_i$.

Later we shall {\em rigorously} show that, under non-arbitrage conditions, such $\alpha_i$ and $\beta_i$ (for $i=0$, ..., $n$) exists and that $q$ given by \eqref{MREDDensity} is indeed a relative entropy minimiser with respect to the prior density $p$ but, firstly, we need some preliminary definitions and results.

We define the cumulant generating functions $c_0$, ..., $c_n$, from $\mathbb{R}$ to $\mathbb{R}\cup\{\infty\}$, by
\begin{equation}
\label{cpi}
c_i(\beta) := \ln \left( \int_{K_{i}}^{K_{i+1}} e^{\beta S} p(S) dS \right). \end{equation}

Notice that $c_i(\beta)<\infty$, for $i < n$ and $\beta\in\mathbb{R}$, since $p\in L^1(K_i, K_{i+1})$ and the exponential function belongs to $L^\infty(K_i, K_{i+1})$.
For $i = n$, the integral is over $[K_n, \infty[$ and we can have $c_n(\beta) = \infty$.
However, $c_n(0)<\infty$ and if $e^{\hat{\beta} S} p(S)$ belongs to $L^1(K_n, \infty)$, then so does $e^{\beta S} p(S)$ for $\beta < \hat{\beta}$.
Therefore, the interior of $c_i$'s effective domain is an interval of the form $]-\infty, \beta^*[$ for some $\beta^*\ge 0$ and, for $i < n$, $\beta^*=\infty$.

\begin{proposition}
\label{Prop:Derivatives}
For $i < n$, $c_i$ is twice differentiable and strictly convex in $]-\infty, \beta^*[$.
Moreover, its first and second derivatives are given by
\begin{equation}
\label{cpi_prime}
c_i'(\beta) = \frac{\int_{K_{i}}^{K_{i+1}} S e^{\beta S} p(S) dS}{\int_{K_{i}}^{K_{i+1}} e^{\beta S} p(S) dS}
\end{equation}
and
\begin{equation}
\label{cpi_2prime}
c_i''(\beta) =
\frac{\int_{K_{i}}^{K_{i+1}} S^2 e^{\beta S} p(S) dS \int_{K_{i}}^{K_{i+1}} e^{\beta S} p(S) dS - \left( \int_{K_{i}}^{K_{i+1}} S e^{\beta S} p(S) dS \right)^2}
{\left( \int_{K_{i}}^{K_{i+1}} e^{\beta S} p(S) dS \right)^2}
\end{equation}
for all $\beta\in\ ]-\infty, \beta^*[$.
\end{proposition}
\begin{proof}
Through standard arguments using the Mean Value Theorem and Lebesgue's Dominated Convergence Theorem one can show that for any $m\in\mathbb{Z}\cap[0, \infty[$, the function
\begin{equation*}
w(\beta) := \int_{K_i}^{K_{i+1}} S^m e^{\beta S}p(S) dS
\end{equation*}
is differentiable in $]-\infty, \beta^*[$ and its derivative can be obtained by differentiating under the integral sign.
The differentiability of $c_i$ and $c_i'$ together with \eqref{cpi_prime} and \eqref{cpi_2prime} follows immediately.

Now we shall prove that $c_i''>0$.
We start by noticing that
\begin{equation*}
\int_{K_i}^{K_{i+1}}\int_{K_i}^{K_{i+1}} (S - R)^2 e^{\beta(S + R)} p(S) p(R) dS dR > 0.
\end{equation*}
Hence,
\begin{equation*}
\frac12 \int_{K_i}^{K_{i+1}}\int_{K_i}^{K_{i+1}} (S^2 + R^2) e^{\beta(S + R)} p(S) p(R) dS dR
>
\int_{K_i}^{K_{i+1}}\int_{K_i}^{K_{i+1}} SR e^{\beta(S + R)} p(S) p(R) dS dR.
\end{equation*}
The left-hand side this last inequality can be rewritten as
\begin{equation*}
\frac12
\int_{K_i}^{K_{i+1}}\int_{K_i}^{K_{i+1}} S^2 e^{\beta(S + R)} p(S) p(R) dS dR
+
\frac12
\int_{K_i}^{K_{i+1}}\int_{K_i}^{K_{i+1}} R^2 e^{\beta(S + R)} p(S) p(R) dS dR,
\end{equation*}
or, simply as
\begin{align*}
\label{eq:lhs}
\int_{K_i}^{K_{i+1}}\int_{K_i}^{K_{i+1}} S^2 e^{\beta(S + R)} p(S) p(R) dS dR
& =
\int_{K_i}^{K_{i+1}}S^2 e^{\beta S} p(S) dS
\int_{K_i}^{K_{i+1}}    e^{\beta R} p(R) dR
\\
& =
\int_{K_i}^{K_{i+1}}S^2 e^{\beta S} p(S) dS
\int_{K_i}^{K_{i+1}}    e^{\beta S} p(S) dS,
\end{align*}
whereas its right-hand side is equal to
\begin{equation*}
\left(\int_{K_i}^{K_{i+1}} S e^{\beta S} p(S) dS\right)
\left(\int_{K_i}^{K_{i+1}} R e^{\beta R} p(R) dR\right)
=
\left(\int_{K_i}^{K_{i+1}} S e^{\beta S} p(S) dS\right)^2.
\end{equation*}
Therefore, the inequality is equivalent to $c_i''>0$.
\end{proof}

Introducing
\begin{equation}
\label{K_i_bar}
\bar{K}_i := \frac{(\tilde{C}_{i} + K_{i} \tilde{D}_{i}) - (\tilde{C}_{i+1} + K_{i+1} \tilde{D}_{i+1})}{\tilde{D}_{i} - \tilde{D}_{i+1}}
\end{equation}
and using \eqref{cpi} we can rewrite \eqref{alpha} and \eqref{beta} in the simpler forms
\begin{align}
c_i'(\beta_i) &= \bar{K}_i,
\label{beta_i}
\\
\alpha_i &= p_i e^{-c_i(\beta_i)}.
\label{alpha_i}
\end{align}
Equation \eqref{beta_i} is easily solved for $\beta_i$ with the Newton-Raphson method using \eqref{cpi_prime} and \eqref{cpi_2prime}.
Once the density $q$ has been obtained in this manner, i.e., $\alpha_i, \beta_i$ have been calculated for $i=0$, ..., $n$,
we can calculate European option prices using numerical integration.

The next results give the existence and uniqueness of such $\beta_i$ and, consequently, $\alpha_i$.

\begin{lemma}
\label{lem:Limits}
Let $m\in\mathbb{Z}\cap[0, \infty[$.
Then, for any $K\in\ ]K_i, K_{i+1}[$ we have
\[
\lim_{\beta\rightarrow\infty}\frac{\int_{K_i}^K S^m e^{\beta S}p(S)dS} {\int_K^{K_{i+1}} S^me^{\beta S}p(S)dS} =
\lim_{\beta\rightarrow-\infty}\frac{\int_K^{K_{i+1}} S^m e^{\beta S}p(S)dS} {\int_{K_i}^K S^me^{\beta S}p(S)dS} =
0.
\]
\end{lemma}
\begin{proof}
We shall consider only the limit when $\beta\rightarrow\infty$ since the other is treated analogously.
Choose $L$ and $M$ such that $K < L < M < K_{i+1}$.
Then, we have
\[
0 \le
\frac{\int_{K_i}^K S^me^{\beta S}p(S)dS} {\int_K^{K_{i+1}} S^me^{\beta S}p(S)dS} \le
\frac{\int_{K_i}^K S^me^{\beta S}p(S)dS} {\int_L^M S^me^{\beta S}p(S)dS}.
\]
Applying the First Mean Value Theorem for Integration yields $S_1\in[K_i, K]$ and $S_2\in [L, M]$ such that
\[
e^{\beta S_1}\int_{K_i}^K S^mp(S)dS =
\int_{K_i}^K S^me^{\beta S}p(S)dS
\quad\text{and}\quad
e^{\beta S_2}\int_L^M S^mp(S)dS =
\int_L^M S^me^{\beta S}p(S)dx.
\]
Therefore
\[
0 \le
\frac{\int_{K_i}^K S^me^{\beta S}p(S)dS} {\int_K^{K_{i+1}} S^me^{\beta S}p(S)dS} \le
\frac{e^{\beta S_1}\int_{K_i}^K S^mp(S)dS} {e^{\beta S_2}\int_L^M S^mp(S)dS} =
Ce^{\beta(S_1 - S_2)},
\]
where $C>0$ does not depend on $\beta$.
Since $S_1 - S_2 \le K - L < 0$, the result follows.
\end{proof}

\begin{proposition}
\label{Prop:LimitsOfc_i'}
We have $\lim_{\beta\rightarrow\infty} c_i'(\beta) = K_{i+1}$ and $\lim_{\beta\rightarrow-\infty} c_i'(\beta) = K_i$.
\end{proposition}
\begin{proof}
Here again, we only consider the first limit since the other is treated analogously.
For all $K\in\ ]K_i, K_{i+1}[$ we have
\[
c_i'(\beta) =
\frac{\int_{K_{i}}^{K_{i+1}} S e^{\beta S} p(S) dS}{\int_{K_{i}}^{K_{i+1}} e^{\beta S} p(S) dS} =
\left[\frac{
  \left(\int_{K_i}^K S e^{\beta S} p(S) dS\right)
  \left(\int_K^{K_{i+1}} S e^{\beta x} p(S) dS\right)^{-1} + 1
}{
  \left(\int_{K_i}^K e^{\beta S} p(S) dS\right)
  \left(\int_K^{K_{i+1}} e^{\beta S} p(S) dS\right)^{-1} + 1
}
\right]
\cdot
\frac{
  \int_K^{K_{i+1}} S e^{\beta S} p(S) dS
}{
  \int_K^{K_{i+1}} e^{\beta S} p(S) dS
}.
\]
Using the previous Lemma, we obtain that the term inside square brackets goes to $1$ as $\beta\rightarrow\infty$.
Now we shall consider the last term above and show that, by choosing a suitable $K$, it is as {\em close} to $K_{i+1}$ as we want.

Firstly, we assume $i < n$.
Then $K_{i+1} < \infty$ and, given a small $\varepsilon>0$, we choose $K = K_{i+1} - \varepsilon$.
The First Mean Value Theorem for Integration gives $S_1\in[K_{i+1} - \varepsilon, K_{i+1}]$ such that
\[
\frac{\int_{K_{i+1} - \varepsilon}^{K_{i+1}} S e^{\beta S} p(S) dS}
{\int_{K_{i+1} - \varepsilon}^{K_{i+1}} e^{\beta S} p(S) dS} =
\frac{S_1\int_{K_{i+1} - \varepsilon}^{K_{i+1}} e^{\beta S} p(S) dS}
{\int_{K_{i+1} - \varepsilon}^{K_{i+1}} e^{\beta S} p(S) dS} =
S_1.
\]
Now, for $i = n$, we have $K_{i+1} = \infty$ and
\[
\frac{\int_K^{\infty} S e^{\beta S} p(S) dS}
{\int_K^{\infty} e^{\beta S} p(S) dS} \ge
\frac{K\int_K^{\infty} e^{\beta S} p(S) dS}
{\int_K^{\infty} e^{\beta S} p(S) dS} =
K.
\]
Hence, the term above goes to $K_{i+1} = \infty$ as $K$ goes to $\infty$.
\end{proof}

\begin{corollary}
\label{Cor:ExistenceOfBeta_i}
For all $i = 0$, ..., $n$, under the non-arbitrage condition $K_i < \bar K_i < K_{i+1}$, where $\bar K_i$ is defined in \eqref{K_i_bar}, there exists a unique solution $\beta_i\in\mathbb{R}$ of \eqref{beta_i}.
\end{corollary}
\begin{proof}
Proposition \ref{Prop:Derivatives} gives that $c_i'$ is continuous and the last Proposition states that $\lim_{\beta\rightarrow-\infty}c_i'(\beta) = K_i$ and $\lim_{\beta\rightarrow\infty}c_i'(\beta) = K_{i+1}$.
Hence the existence follows from the Intermediate Value Theorem.
Additionally, Proposition \ref{Prop:Derivatives} also gives that $c_i'$ is strictly increasing and the uniqueness follows.
\end{proof}

\begin{theorem}
\label{Thm:Minimiser}
For all $i = 0$, ..., $n$, let $\alpha_i$ and $\beta_i$ be defined by equations \eqref{beta_i} and \eqref{alpha_i}.
Then $q : [0, \infty[ \rightarrow {\mathbb R}$ given by
\begin{equation*}
q(S) = \alpha_i e^{\beta_i S} p(S) \quad \forall S \in [K_i, K_{i+1}[
\end{equation*}
minimises $H(q || p)$.
\end{theorem}
\begin{proof}
This is shown just as Theorem 2.6 in \cite{NeriSchneider2009} by using Theorem 2.5 by Csisz\'ar also stated there.
\end{proof}

\subsection{The Special Case of Non-Relative Entropy}
\label{ss:LebesguePrior}

The non-relative entropy can be seen as special case of the relative entropy for which no prior $p$ is given or, roughly speaking, the prior is given by Lebesgue-measure $p \equiv 1$.
Then equation \eqref{cpi} reduces to the following analytic expression:
\begin{equation*}
c_i(\beta) = \left\{
\begin{array}{cl}
\displaystyle \ln \left( \frac{e^{\beta K_{i+1}} - e^{\beta K_i}}{\beta} \right) \quad & \text{for } i < n \text{ and } \beta \neq 0, \\
\\
\displaystyle \ln(K_{i+1} - K_i) & \text{for } i < n \text{ and } \beta = 0,\\
\\
\displaystyle \ln\left(-\frac{e^{\beta K_i}}{\beta}\right) & \text{for } i = n \text{ and } \beta < 0,
\end{array}
\right.
\end{equation*}
and the first and second derivatives \eqref{cpi_prime} and \eqref{cpi_2prime} reduce to
\begin{equation*}
c_i'(\beta) = \left\{
\begin{array}{cl}
\displaystyle \frac{K_{i+1} e^{\beta K_{i+1}} - K_i e^{\beta K_i}}{e^{\beta K_{i+1}} - e^{\beta K_i}} - \frac{1}{\beta} \quad & \text{for } i < n \text{ and } \beta \neq 0, \\
\\
\displaystyle \frac{K_{i+1} + K_i}{2} & \text{for } i < n \text{ and } \beta = 0, \\
\\
\displaystyle K_i - \frac{1}{\beta} \quad & \text{for } i = n \text{ and } \beta < 0,
\end{array}
\right.
\end{equation*}
and
\begin{equation*}
c_i''(\beta) = \left\{
\begin{array}{cl}
- (K_{i+1} - K_i)^2
  \dfrac{e^{\beta(K_{i+1} + K_i)}}{(e^{\beta K_{i+1}} - e^{\beta K_i})^2}
+ \dfrac1{\beta^2}
\quad & \text{for } i < n \text{ and } \beta \neq 0, \\
\\
\dfrac{(K_{i+1} - K_i)^2}{12} & \text{for } i < n \text{ and } \beta = 0, \\
\\
\dfrac1{\beta^2} \quad & \text{for } i = n \text{ and } \beta < 0.
\end{array}
\right.
\end{equation*}
Using these expressions instead of \eqref{cpi}, \eqref{cpi_prime}, \eqref{cpi_2prime}
allows for numerical integration to be avoided in an implementation in this case.

\section{Minimizer Matching Call Option Prices}
\label{s:MRED_Call}

Buchen and Kelly describe in \cite{BuchenKelly1996} a multi-dimensional Newton-Raphson algorithm to find the maximum entropy distribution (MED)
if only call prices are given as constraints. The entropy $H$ of a probability density $q$ over $[0, \infty[$ is given by
\begin{equation*}
H(q) = -\int_0^\infty q(S) \ln q(S) dS.
\end{equation*}
The minus sign in the definition ensures that $H$ is always positive for discrete densities (where the integral sign is replaced by a sum in the definition).
For continuous densities, $H$ is usually, but not always, positive. For example, the uniform density $q(S) \equiv u$ over the interval $[0, u^{-1}]$
has negative entropy $H(q) = -\int_0^{u^{-1}} u \ln u \; dx = -\ln u < 0$ for $u>1$.

In \cite{NeriSchneider2011}, we show how the results of \cite{NeriSchneider2009} together with the Legendre transform can be applied to
obtain a fast and more robust Newton-Raphson algorithm to calculate the Buchen-Kelly MED.
The main reason the algorithm is more stable is that the Hessian matrix has a very simple tridiagonal form.
In section II.A of their paper, Buchen and Kelly also consider the case of ``minimum cross entropy'' (which we call relative entropy here) for a given prior density.

We now show how essentially the same algorithm as that in \cite{NeriSchneider2011} can be applied to the relative entropy case.
The next proposition consolidates and generalises the results of section 4 of \cite{NeriSchneider2011}, describing the entropy $H$, the gradient vector and the Hessian matrix to the case in which a prior density $p$ is given.

Arbitrage free digital prices must lie between left and right call-spread prices, i.e.,
\begin{equation}
\label{NoArbitrageDigitals}
-\frac{\tilde{C}_{i} - \tilde{C}_{i-1}}{K_i - K_{i-1}} > \tilde{D}_i > -\frac{\tilde{C}_{i+1} - \tilde{C}_i}{K_{i+1} - K_i}, \qquad \forall i = 1, ..., n,
\end{equation}
where the rightmost quantity for $i = n$ must be read as zero.

We introduce the set $\Omega \subset \mathbb{R}^n$ of all $\tilde{D} = (\tilde{D}_1, ..., \tilde{D}_n) \in \mathbb{R}^n$ verifying \eqref{NoArbitrageDigitals}.
Note that $\Omega$ is an open $n$-dimensional rectangle.
Define $q_{\tilde{D}}$ as the density obtained as in Theorem \ref{Thm:Minimiser} for given (undiscounted) digital prices $\tilde{D}$.

\begin{proposition}
\label{Prop:Summary}
For all $\tilde{D} \in \Omega$ the relative entropy of $q_{\tilde{D}}$ with respect to $p$ can be expressed as
\begin{equation*}
H(q_{\tilde{D}} \| p) = \sum_{i=0}^n p_i \ln p_i + \sum_{i=0}^n p_i c_i^*(\bar{K}_i),
\end{equation*}
where $c_i^*$ is the Legendre transform of $c_i$,
$p_i := \tilde{D}_{i} - \tilde{D}_{i+1}$ and $\bar{K}_i$ is given by \eqref{K_i_bar}.

As a function of digital prices, $H:\Omega\rightarrow \mathbb{R}$ is strictly convex, twice differentiable and, for all $\tilde{D}\in\Omega$, we have
\begin{equation*}
\frac{\partial H}{\partial \tilde{D}_i}(\tilde{D}) = \ln g_{\tilde{D}}(K_i+) - \ln g_{\tilde{D}}(K_i-), \quad \forall i = 1, ..., n,
\end{equation*}
where $g_{\tilde{D}} := q_{\tilde{D}}/p$ and
\begin{equation*}
g_{\tilde{D}}(K_i-) := \lim_{S\rightarrow K_i^-} g_{\tilde{D}}(S) = \alpha_{i-1} e^{\beta_{i-1} K_i}
\quad \text{and} \quad
g_{\tilde{D}}(K_i+) := \lim_{S\rightarrow K_i^+} g_{\tilde{D}}(S) = \alpha_i e^{\beta_i K_i},
\end{equation*}
with $\alpha_i$ and $\beta_i$ given by \eqref{beta_i} and \eqref{alpha_i}, for all $i=0, ..., n$.

In addition, the Hessian matrix at any $\tilde{D}\in\Omega$ is symmetric and tridiagonal with entries given by
\begin{align*}
\dfrac{\partial^2 H}{\partial \tilde{D}_i^2}(\tilde{D})
&=
\dfrac{1}{p_{i-1}}
+ \dfrac{1}{p_i}
+ \dfrac{(K_i - \bar{K}_{i-1})^2}{p_{i-1} c_{i-1}''(\beta_{i-1})}
+ \dfrac{(\bar{K}_i - K_i)^2}{p_i c_i''(\beta_i)}, \quad &
\forall i &= 1, ..., n, \\
\\
\dfrac{\partial^2 H}{\partial \tilde{D}_i \partial \tilde{D}_{i+1}}(\tilde{D}) &=
-\dfrac{1}{p_i}
+ \dfrac{(\bar{K}_i - K_i)(K_{i+1}-\bar{K}_i)}{p_i c_i''(\beta_i)}, \quad &
\forall i &= 1, ..., n - 1,
\end{align*}
\end{proposition}

\begin{proof}
Let $g_{\tilde{D}} := q_{\tilde{D}}/p$ be the piecewise-exponential Radon-Nikod\'ym derivative given in Theorem \ref{Thm:Minimiser}. We have
\begin{equation*}
H(q_{\tilde{D}} \| p) = \int_{K_i}^{K_{i+1}} q_{\tilde{D}}(S) \ln \frac{q_{\tilde{D}}(S)}{p(S)} dS =
\ln \alpha_i \int_{K_i}^{K_{i+1}} q_{\tilde{D}}(S) dS + \beta_i \int_{K_i}^{K_{i+1}} S q_{\tilde{D}}(S) dS,
\end{equation*}
since $g_{\tilde{D}}(S) = \alpha_i e^{\beta_i S}$ on $[K_i, K_{i+1}[$.
Then using \eqref{beta_i} and \eqref{alpha_i} the proof goes through as the proofs of Theorem 4.1, Theorem 4.2, Lemma 5.1 and Proposition 5.2 from \cite{NeriSchneider2011} given there by using the generalised versions of $c_i$ and $c_i^*$ introduced above and observing the absence of the minus sign in the definition of relative entropy.
\end{proof}

Notice that $p_i$ and $\bar{K}_i$ are given purely in terms of option prices, for all $i=0, ..., n$, and so is $H(q_{\tilde{D}} \| p)$.
Notice also that if the prior density $p$ already matches the call prices, then $\alpha_i = 1$ and $\beta_i = 0$ for all $i=0, ..., n$.
From the relationship $c_i^*(\bar{K}_i) = \beta_i \bar{K}_i - c_i(\beta_i)$, it follows that, in this case, $c_i^*(\bar{K}_i) = -c_i(\beta_i)$.
Since $\ln p_i = \ln \alpha_i + c_i(\beta_i) = -c_i^*(\bar{K}_i)$, the proposition above gives that $H(q_{\tilde{D}} \| p) = 0$, as expected.

The expression for the derivative of $H$ gives that if $\tilde{D}$ minimises $H$ (i.e., $\tilde{D}$ is a root of the gradient of $H$), then $g_{\tilde{D}}$ is continuous.
Furthermore, the MRED $q_{\tilde{D}} = g_{\tilde{D}}p$ has the same points of discontinuity as $p$.

Using these last results, essentially the same Newton-Raphson algorithm as in \cite{NeriSchneider2011} can be applied to find the relative entropy minimiser $q$.
The only differences are that the functions $c_0''$, ..., $c_n''$ must be replaced by their relative entropy versions \eqref{cpi_2prime} in the Hessian matrix of Proposition \ref{Prop:Summary}, and that in each iteration step, for a given set of digital prices, the algorithm of Section \ref{s:MRED_CallDigital} must be used to calculate the MRED, instead of its non-relative version.

\section{Probability Densities for Characteristic Function Models}
\label{s:CharacteristicFunctionModels}

In this section, we look at four models that are popular in equity derivatives pricing.
Our aim is to use the densities they give for the stock price at a fixed maturity as prior densities.
In two of the models we have chosen, the Black-Scholes model and the Variance Gamma model,
the density is analytically available.
In the other two, the Heston and the Sch\"obel-Zhu stochastic volatility models, it is not.
We therefore give a brief overview of these models and show how to calculate their densities in each case.

Let $\tilde{p}$ be the density of $x(T) := \ln S(T)$.
To simplify notation, we will usually write $x$ and $S$, respectively, when it is clear from the context that we have fixed the maturity $T$.
Then the density $p$ of $S$ itself is given by
\begin{equation*}
p(S) = \frac{1}{S}\ \tilde{p}(\ln S),
\end{equation*}
since $\int_{-\infty}^{\ln a} \tilde{p}(x) dx = \int_0^a \tilde{p}(\ln S) \frac{1}{S} dS = \int_0^a p(S) dS$ by change-of-variables formula.

If the characteristic function $\phi$ of $p$, given by
\begin{equation}
\label{CharacteristicFunction}
\phi(u) := {\mathbb E} \left[ e^{i u x} \right] = \int_{-\infty}^\infty e^{i u x} \tilde{p}(x) dx,
\end{equation}
is known,
as in the Heston \cite{Heston1993} or Sch\"{o}bel-Zhu \cite{SchoebelZhu1999} stochastic volatility models,
then $p$ can be obtained via Fourier inversion:
\begin{equation*}
\tilde{p}(x) = \frac{1}{2 \pi} \int_{-\infty}^\infty e^{-i u x} \phi(u) du.
\end{equation*}
Since $\tilde{p}$ is a real-valued function, it follows from \eqref{CharacteristicFunction} that $\phi(-u) = \overline{\phi(u)}$,
and we have
\begin{align}
\tilde{p}(x)
&= \frac{1}{2 \pi} \int_0^\infty e^{-i u x} \phi(u) du + \frac{1}{2 \pi} \int_0^\infty e^{i u x} \phi(-u) du
\nonumber
\\
&= \frac{1}{\pi} \int_0^\infty \Re \left[ e^{-i u x} \phi(u) \right] du,
\label{CFMDensity}
\end{align}
where $\Re [z] = (z + \overline{z}) / 2$ denotes the real part of a complex number $z$.
It can immediately be seen that an anti-derivative of $\tilde{p}$ is given by
\begin{equation*}
\tilde{P}_0(x) = -\frac{1}{2 \pi} \int_{-\infty}^\infty \frac{e^{-i u x}}{i u} \phi(u) du.
\end{equation*}
Furthermore, it can be shown in a similar way as the Fourier Inversion Theorem itself, that
$\lim_{x \to -\infty} \tilde{P}_0(x) = -\frac{1}{2}$ and $\lim_{x \to \infty} \tilde{P}_0(x) = \frac{1}{2}$,
and therefore the function
\begin{align*}
\tilde{P}(x)
&= \frac{1}{2} - \frac{1}{2 \pi} \int_{-\infty}^\infty \frac{e^{-i u x}}{i u} \phi(u) du
\\
&= \frac{1}{2} - \frac{1}{\pi} \int_0^\infty \Re \left[ \frac{e^{-i u x} \phi(u)}{i u} \right] du
\end{align*}
gives an expression for the distribution function.

For pricing, we use the general formulation of Bakshi and Madan \cite{BakshiMadan2000}.
This can be used for a large class of characteristic function models that contains the Heston, Sch\"{o}bel-Zhu and Variance Gamma models
(see section 2 in \cite{BakshiMadan2000}, in particular Case 2 on p.218).
We have $S > K$ if, and only if, $x > \ln K$. Let
\begin{align}
\label{Pi1}
\Pi_1 &:= 1 - \tilde{P}_S(\ln K) = \frac{1}{2} + \frac{1}{\pi} \int_0^\infty \Re \left[ \frac{e^{-i u \ln K} \phi(u - i)}{i u \phi(-i)} \right] du,
\\
\label{Pi2}
\Pi_2 &:= 1 - \tilde{P}(\ln K) = \frac{1}{2} + \frac{1}{\pi} \int_0^\infty \Re \left[ \frac{e^{-i u \ln K} \phi(u)}{i u} \right] du,
\end{align}
represent the probabilities of $S$ finishing in-the-money at time $T$ in case the stock $S$ itself or a risk-free bond is used as num\'{e}raire, respectively.
From \eqref{CharacteristicFunction} we can see that $\phi(-i) = {\mathbb E} \left[ e^x \right] = {\mathbb E} \left[ S \right]$,
so that the quotient
\begin{equation*}
\frac{\phi(u - i)}{\phi(-i)} = \int_{-\infty}^\infty e^{i u x} \frac{S}{{\mathbb E} \left[ S \right]} \tilde{p}(x) dx
\end{equation*}
contains the appropriate change of measure.

The price $C$ of a European call option on a stock paying a dividend yield $d$ is then obtained through the formula
\begin{equation}
C = e^{-dT} S \Pi_1 - e^{-rT} K \Pi_2,
\end{equation}
and the price $D$ of a European digital call option prices through
\begin{equation}
D = e^{-rT} \Pi_2,
\end{equation}
where $r$ is the risk-free, continuously-compounded interest rate.

The integrals in \eqref{Pi1}, \eqref{Pi2} must of course be truncated at some point $a$, which depends on the decay of the characteristic function of the model considered.

\subsection{The Black-Scholes Model}
\label{ss:BlackScholesModel}

Let the parameters $r$, $d$ and $T$ be given as above, and let $\sigma > 0$ be the volatility of the stock price.
In the Black-Scholes model, the logarithm $x(t) := \ln S(t)$ follows the SDE
\begin{equation*}
d x(t) = \left(r - d - \frac{1}{2} \sigma^2\right) dt + \sigma dW(t)
\end{equation*}
Define $\mu := \ln S(0) + \left(r - d - \frac{1}{2} \sigma^2\right) T$.
The density of $x(T)$ is normal and given by
\begin{equation}
\label{Black-Scholes Density}
\tilde{p}(x) = \frac{1}{\sqrt{2 \pi \sigma^2 T}} e^{-\frac{(x - \mu)^2}{2 \sigma^2 T}}.
\end{equation}
The characteristic function of $\tilde{p}$ has a very simple form and is given by
\begin{equation}
\label{Black-Scholes CF}
\phi(u) = e^{i u \mu - \frac{1}{2} \sigma^2 u^2 T}.
\end{equation}

Of course it is faster to use \eqref{Black-Scholes Density} directly instead of \eqref{Black-Scholes CF} and \eqref{CFMDensity},
but comparing these two methods lets one measure the additional computational burden.

\subsection{The Heston Model}
\label{ss:HestonModel}

One of the most popular models for derivative pricing in equity and FX markets is the stochastic volatility model introduced by Heston \cite{Heston1993}.
Let $x(t) := \ln S(t)$. The model is given, in the risk-neutral measure, by the following two SDE's:
\begin{align}
d x(t) &= (r - d - \frac{1}{2} v(t)) dt + \sqrt{v(t)} dW_1(t),
\\
d v(t) &= (\kappa \theta - (\kappa + \lambda) v(t)) dt + \sigma \sqrt{v(t)} dW_2(t),
\end{align}
where $\langle dW_1(t), dW_2(t) \rangle = \rho dt$.
The variance rate $v$ follows a Cox-Ingersoll-Ross square-root process \cite{CoxIngersollRoss1985}.

The parameter $\lambda$ represents the market price of volatility risk.
Since we are only interested in pricing, we always set $\lambda = 0$ in what follows (see Gatheral \cite{Gatheral2006}, chapter 2).

Heston calculates the characteristic function solution, but as pointed out in \cite{SchoebelZhu1999}, there is a (now well-known) issue
when taking the complex logarithm. To be clear, we therefore give the formulation of the characteristic function that we use.

Define
\begin{equation}
\label{negativeSquareRoot}
b := \kappa + \lambda, \quad
d_1 := \sqrt{(i \rho \sigma u - b)^2 + \sigma^2 \phi(i + u)},
\quad d_2 := -d_1
\quad\text{and}\quad
g := \frac{b - i \rho \sigma u + d_2}{b - i \rho \sigma u - d_2}.
\end{equation}
Introducing
\begin{align*}
C &:= (r - d) u i T + \frac{\kappa \theta}{\sigma^2} \left( (b - i \rho \sigma u + d_2) T + 2 \ln \frac{1 - g e^{d_2 T}}{1 - g} \right),
\\
D &:= \frac{b - i \rho \sigma u + d_2}{\sigma^2} \frac{1 - e^{d_2 T}}{1 - g e^{d_2 T}},
\end{align*}
the characteristic function of $x(T)$ is then given as
\begin{equation}
\phi(u) = e^{C + D v_0 + i u \ln S(0)}.
\end{equation}
Since implementations of the complex square-root usually return the root with non-negative real part ($d_1$),
the key is simply to take the other root ($d_2$), as is done in equation \eqref{negativeSquareRoot}.
As shown in \cite{AlbrecherMayerSchoutensTistaert2007} and \cite{LordKahl2010}, this takes care of the whole issue.

\subsection{The Sch\"{o}bel-Zhu Model}
\label{ss:SchoebelZhuModel}

The Sch\"{o}bel-Zhu model \cite{SchoebelZhu1999} is an extension of the Stein and Stein stochastic volatility model
\cite{SteinStein1991} with correlation $\rho \neq 0$ allowed (see also \cite{Clark2011}, \cite{Zhu2010}).
It is described, in the risk-neutral measure, by the following two SDE's:
\begin{align}
d x(t) &= (r - d - \frac{1}{2} v^2(t)) dt + v(t) dW_1(t),
\\
d v(t) &= \kappa (\theta - v(t)) dt + \sigma dW_2(t),
\end{align}
where $\langle dW_1(t), dW_2(t) \rangle = \rho dt$.
The volatility $v$ follows a mean-reverting Ornstein-Uhlenbeck process.

The characteristic function for this model is given by Sch\"{o}bel and Zhu in \cite{SchoebelZhu1999}.
As Lord and Kahl point out (\cite{LordKahl2010}, section 4.2), similar attention has to be paid when taking the complex logarithm
in this model's characteristic function as in Heston's. By directly relating the two characteristic functions (eq. 4.14),
they show how the Sch\"{o}bel-Zhu model can also be implemented safely. In the case study in \ref{ss:SPXMarket} presented in the following section
with SPX option data and a maturity of less than half a year, however, we observed no problems with the characteristic function originally proposed by Sch\"{o}bel and Zhu.

\subsection{The Variance Gamma Model}
\label{ss:VarianceGammaModel}

The Variance Gamma (VG) process was introduced in \cite{MadanCarrChang1998}, \cite{MadanSeneta1990}.
The density for $x = \ln S(T)$ is given explicitly in Theorem 1 in \cite{MadanCarrChang1998}.
Define
\begin{equation*}
\omega := \frac{1}{\nu} \ln (1 - \theta \nu - \frac{1}{2} \sigma^2 \nu)
\quad \text{and} \quad
\tilde{x} = x - \ln S(0) - (r - d + \omega) T.
\end{equation*}
Then the density $\tilde{p}$ is given by
\begin{equation}
\label{Variance Gamma Density}
\tilde{p}(x) = \frac{2 \exp(\theta \tilde{x} / \sigma^2)}{\nu^{T / \nu} \sqrt{2 \pi} \sigma \Gamma(\frac{T}{\nu})}
\cdot
\left( \frac{\tilde{x}^2}{\frac{2 \sigma^2}{\nu} + \theta^2} \right)^{\frac{T}{2 \nu} - \frac{1}{4}}
\cdot
K_{\frac{T}{\nu} - \frac{1}{2}} \left( \frac{1}{\sigma^2} \sqrt{\tilde{x}^2 \left( \frac{2 \sigma^2}{\nu} + \theta^2 \right) } \right),
\end{equation}
where $\Gamma$ is the Gamma-function and $K$ is the modified Bessel function of the second kind.

If the parameter $\nu$ is set to zero, the characteristic function of $p$ reduces to the Black-Scholes one given in \eqref{Black-Scholes CF}.
Otherwise, it is given by
\begin{equation}
\label{Variance Gamma CF}
\phi(u) = e^{i u \ln(S(0)) + (r - d + \omega) T} \left( 1 - i \theta \nu u + \frac{1}{2} \sigma^2 u^2 \nu \right)^{-\frac{T}{\nu}}.
\end{equation}
Lord and Kahl show (\cite{LordKahl2010}, section 4.1) that this formulation of the characteristic function is safe.

Again, as with the Black-Scholes model, since both the density and the characteristic function are available, it is possible to compare the two different methods
\eqref{Variance Gamma Density} vs. \eqref{Variance Gamma CF} and \eqref{CFMDensity}.

\section{Two Numerical Examples}
\label{s:NumericalExamples}

\subsection{A Fictitious Market and Black-Scholes and Heston Prior Densities}
\label{ss:BlackScholesMarket}

In our first example, we take a hypothetical market with $r=d=0, T=1, S=F=100,$
in which call option prices are given by the Black-Scholes formula with volatility $\sigma = 0.25$.
As prior densities, we take
\begin{itemize}
\item
$p_{BS}$, a Black-Scholes log-normal density, but this time with volatility $\sigma_p = 0.20$,
\item
$p_H$, a Heston density, with parameters $\kappa = 1, \theta = 0.04, \rho = -0.3, \sigma = 0.25, v_0 = 0.04$,
which leads to implied volatilities of $0.2418,	0.2125,	0.1923,	0.1855,	0.1884$ at strikes
$60, 80, 100, 120, 140$, respectively.
\end{itemize}

We calculate the Buchen-Kelly MED (Lebesgue prior), using the algorithm presented in \cite{NeriSchneider2011},
and the two Buchen-Kelly MREDs with priors $p_{BS}$ and $p_H$, using the generalised algorithm presented in section \ref{s:MRED_Call},
and compare the resulting call and digital option prices to see the influence of the priors.

The different call and digital option prices are reported in table \ref{tab:CallAndDigitalPrices}.
The densities were calculated using call prices at strikes
\begin{itemize}
\item
$K_0 = 0, K_1 = 100$
\item
$K_0 = 0, K_1 = 60, K_2 = 100, K_3 = 140$
\item
$K_0 = 0, K_1 = 60, K_2 = 80, K_3 = 100, K_4 = 120, K_5 = 140$
\end{itemize}
as respective constraints. These strikes are the ones in boldface in table \ref{tab:CallAndDigitalPrices}.

\begin{table}[htbp]
  \footnotesize
  \centering
  \caption{Comparison of Call and Digital Prices}
    \begin{tabular}{c|ccc|ccc}
    \addlinespace
    \toprule
    & & Call Prices & & & Digital Prices & \\
    Strike & MED BK & MRED BS & MRED Heston & MED BK & MRED BS & MRED Heston \\
    \midrule
    20    & 80.0538 & 80.0000 & 80.0000 & 0.9936 & 1.0000 & 1.0000 \\
    40    & 60.3244 & 60.0000 & 60.0094 & 0.9766 & 1.0000 & 0.9979 \\
    60    & 41.1698 & 40.0637 & 40.3043 & 0.9316 & 0.9841 & 0.9595 \\
    80    & 23.5389 & 21.9716 & 22.5717 & 0.8124 & 0.7758 & 0.7828 \\
    {\bf 100} & 9.9476 & 9.9476 & 9.9476 & 0.4962 & 0.4420 & 0.4763 \\
    120   & 3.6684 & 3.6071 & 3.3294 & 0.1830 & 0.2039 & 0.1977 \\
    140   & 1.3528 & 1.0596 & 1.0051 & 0.0675 & 0.0693 & 0.0593 \\
    160   & 0.4989 & 0.2688 & 0.3239 & 0.0249 & 0.0192 & 0.0175 \\
    180   & 0.1840 & 0.0621 & 0.1171 & 0.0092 & 0.0047 & 0.0056 \\
    \midrule
          &       &       &       &       &       &  \\
    20    & 80.0000 & 80.0000 & 80.0000 & 1.0000 & 1.0000 & 1.0000 \\
    40    & 60.0015 & 60.0003 & 60.0012 & 0.9997 & 0.9998 & 0.9996 \\
    {\bf 60} & 40.1454 & 40.1454 & 40.1454 & 0.9669 & 0.9753 & 0.9715 \\
    80    & 22.5812 & 22.0890 & 22.3433 & 0.7743 & 0.7818 & 0.7770 \\
    {\bf 100} & 9.9476 & 9.9476 & 9.9476 & 0.4646 & 0.4424 & 0.4633 \\
    120   & 3.7041 & 3.7051 & 3.5189 & 0.1945 & 0.1976 & 0.1926 \\
    {\bf 140} & 1.2139 & 1.2139 & 1.2139 & 0.0705 & 0.0707 & 0.0617 \\
    160   & 0.3800 & 0.3569 & 0.4669 & 0.0221 & 0.0227 & 0.0207 \\
    180   & 0.1190 & 0.0961 & 0.2067 & 0.0069 & 0.0065 & 0.0077 \\
    \midrule
          &       &       &       &       &       &  \\
    20    & 80.0001 & 80.0000 & 80.0000 & 1.0000 & 1.0000 & 1.0000 \\
    40    & 60.0033 & 60.0002 & 60.0014 & 0.9994 & 0.9999 & 0.9996 \\
    {\bf 60} & 40.1454 & 40.1454 & 40.1454 & 0.9726 & 0.9727 & 0.9726 \\
    {\bf 80} & 22.2656 & 22.2656 & 22.2656 & 0.7794 & 0.7781 & 0.7804 \\
    {\bf 100} & 9.9476 & 9.9476 & 9.9476 & 0.4510 & 0.4499 & 0.4510 \\
    {\bf 120} & 3.7059 & 3.7059 & 3.7059 & 0.1971 & 0.1961 & 0.1958 \\
    {\bf 140} & 1.2139 & 1.2139 & 1.2139 & 0.0700 & 0.0711 & 0.0689 \\
    160   & 0.3834 & 0.3545 & 0.4105 & 0.0221 & 0.0227 & 0.0211 \\
    180   & 0.1211 & 0.0948 & 0.1564 & 0.0070 & 0.0064 & 0.0071 \\
    \bottomrule
    \end{tabular}
  \label{tab:CallAndDigitalPrices}
\end{table}

We see that in the first case, where we had only the forward and an at-the-money call option as constraints,
there are significant differences in both call and digital prices.
The presence of the log-normal prior density makes MRED BS call prices cheaper compared to the original BS prices.
Of course, under the prior density itself, call prices were cheaper because of the lower volatility, and this effect seems to persist.
We also see that the fatter (exponential) tails of the MED translate into higher prices of deeply
in- or out-of-the-money call options when compared to the other two densities.

However, as we add call prices at more strikes as constraints, the differences become smaller.
In the third part of table \ref{tab:CallAndDigitalPrices}, the prices of both call and digital options are clearly converging.

\subsection{SPX Option Prices and Heston, Sch\"{o}bel-Zhu and VG Prior Densities}
\label{ss:SPXMarket}

In our second example, we look at call (ticker symbol SPX) and digital (BSZ) options on the Standard and Poor's $500$ stock index \cite{CBOE2010}.
The market data is from 18 July 2011. We consider those options which expire on 17 December 2011 and calibrate a Heston, Sch\"{o}bel-Zhu and VG model,
respectively, to call prices for 15 strikes $900, 950, ... 1550, 1600$ at constant intervals of $50$ using a Levenberg-Marquardt least-squares method
(\cite{PTVF2002}, \cite{Clark2011}).
The model parameters we obtained are given in table \ref{tab:ModelParameters}.
\begin{table}[htbp]
  \centering
  \caption{Model Parameters}
    \begin{tabular}{cccc}
    \addlinespace
    \toprule
    Parameter & Heston & SZ & VG \\
    \midrule
    $\kappa$ & 0.8568  & 1.6316  & n/a     \\
    $\theta$ & 0.0800  & 0.1731  & -0.2808 \\
    $\rho$   & -0.8016 & -0.8031 & n/a     \\
    $\sigma$ & 0.5473  & 0.3249  & 0.1535  \\
    $v_0$    & 0.0421  & 0.1887  & n/a     \\
    $\nu$    & n/a     & n/a     & 0.3638  \\
    \bottomrule
    \end{tabular}
  \label{tab:ModelParameters}
\end{table}

Figure \ref{Fig:1} shows the market implied volatility skew and the volatility skews generated by the three models.
Apart from the last two strikes at $1550$ and $1600$, the fit looks quite good in all three cases:
\begin{figure}[ht]
\centering
\includegraphics[width=\textwidth]{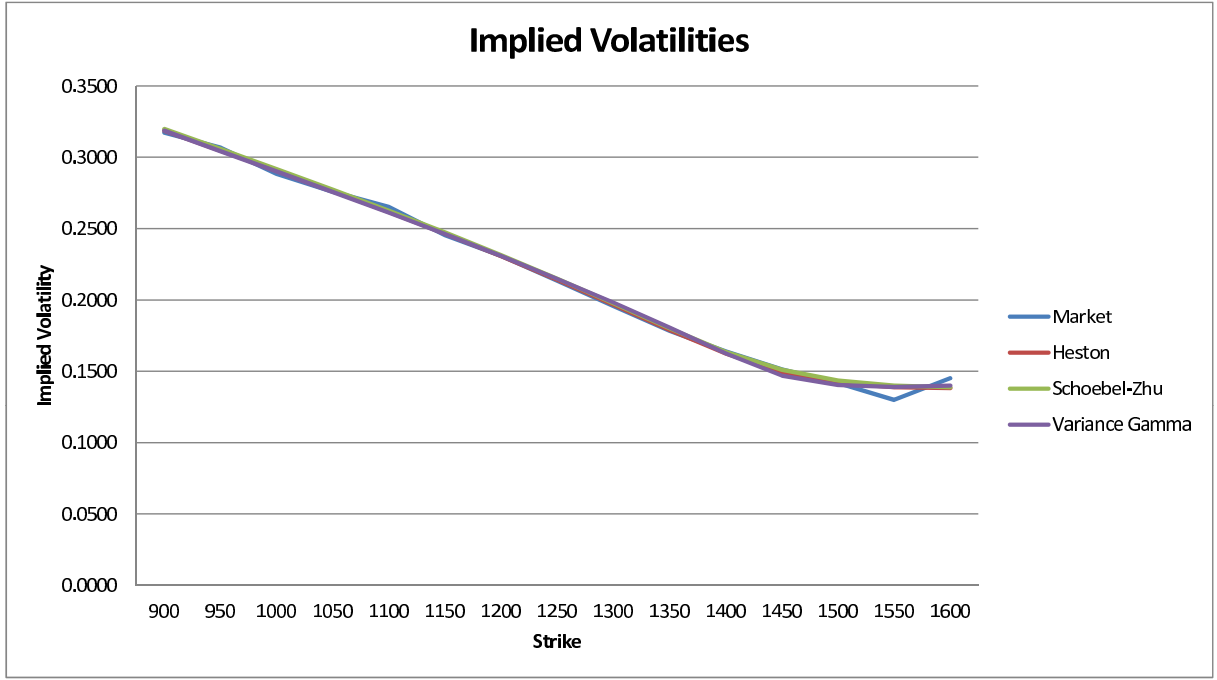}
\caption{Graphs of the four volatility skews}
\label{Fig:1}
\end{figure}

Using formulas for the densities directly, if available, or otherwise numerical inversion \eqref{CFMDensity},
we plot the densities for $S(T)$ given by these models in figure \ref{Fig:2}:
\begin{figure}[ht]
\centering
\includegraphics[width=\textwidth]{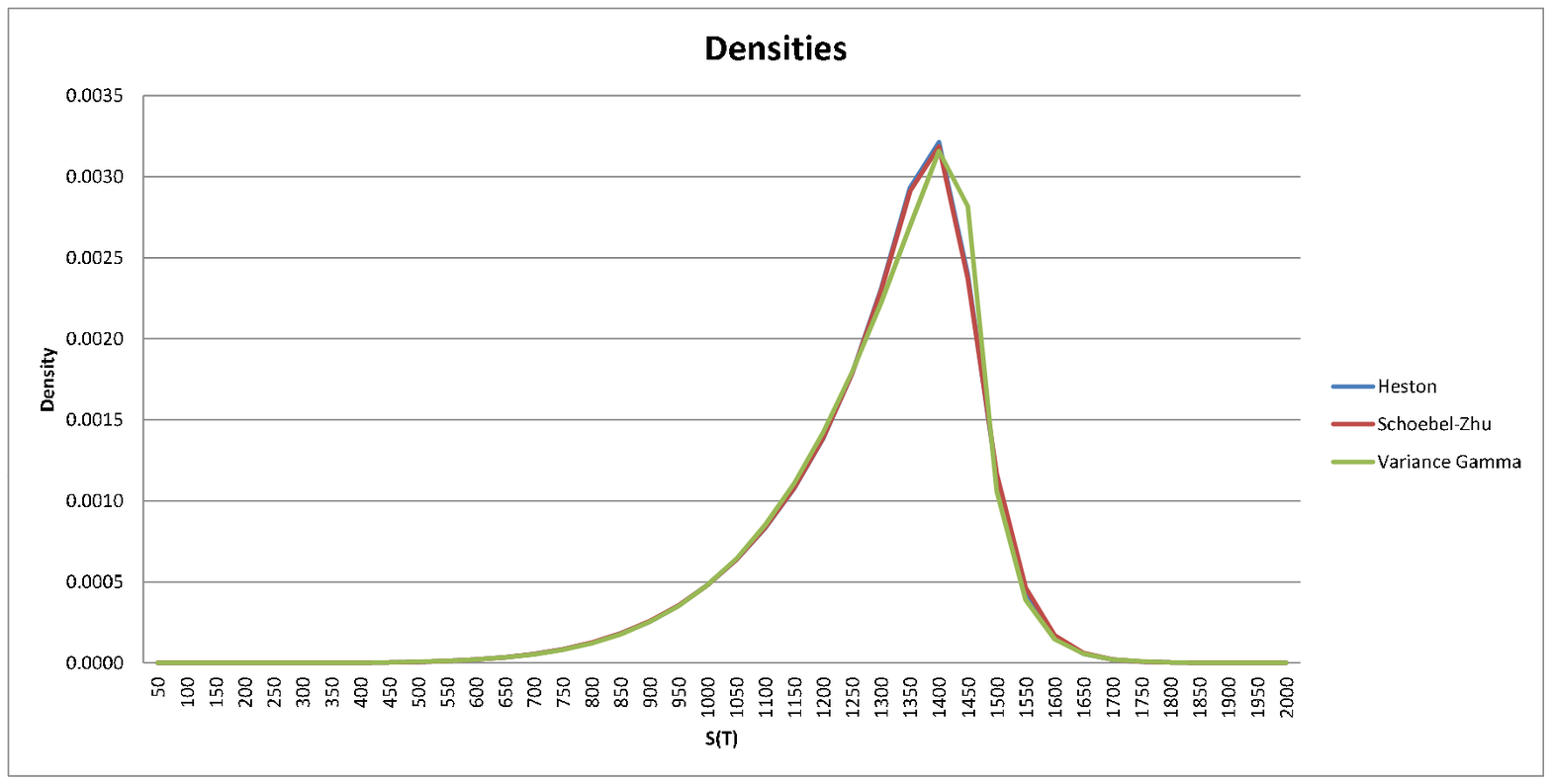}
\caption{Graphs of the three model densities}
\label{Fig:2}
\end{figure}
The Heston and Sch\"{o}bel-Zhu densities are almost indistinguishable from one another,
whereas the VG density has a somewhat different shape with a slightly thinner right tail.

\begin{table}[htbp]
  \centering
  \caption{Squared Errors and Relative Entropy $H(h \| p)$}
    \begin{tabular}{cccc}
    \addlinespace
    \toprule
    Model: & Heston & SZ    & VG \\
    \midrule
    $\sum (\sigma_i - \hat{\sigma}_i)^2$: & 1.59E-04 & 1.86E-04 & 1.74E-04 \\
    Relative Entropy: & 0.1311 & 0.1315 & 0.1374 \\
    \bottomrule
    \end{tabular}
  \label{tab:RelativeEntropy}
\end{table}

Table \ref{tab:RelativeEntropy} shows the sum of squared errors $\sum_{i=1}^{15} (\sigma_i^{SPX} - \hat{\sigma}_i^{model})^2$ between the market (SPX)
implied volatilities and the model implied volatilities.
The relative entropy $H(q \| p)$ can be seen as an alternative measure of fit, since by \eqref{DefinitionOfRelativeEntropyForDensities} it measures
how much the prior density $p$ needs to be deformed to obtain a density $q$ that perfectly matches the given market data.
Interestingly, the Heston model fits best under either criterion, but the order of the Sch\"{o}bel-Zhu and VG models is reversed in the two cases.

\begin{table}[htbp]
  \scriptsize
  \centering
  \caption{Digital Prices}
    \begin{tabular}{cccccccccc}
    \addlinespace
    \toprule
    {\bf Strike} & Market (Mid) & Call Spreads & MED BK & Heston & MRED Heston & SZ & MRED SZ & VG & MRED VG \\
    \midrule
    900   & 0.9500 & 0.9560 & 0.9609 & 0.9682 & 0.9612 & 0.9678 & 0.9612 & 0.9687 & 0.9607 \\
    950   & 0.9500 & 0.9540 & 0.9544 & 0.9532 & 0.9544 & 0.9527 & 0.9544 & 0.9537 & 0.9544 \\
    1000  & 0.9400 & 0.9335 & 0.9432 & 0.9325 & 0.9433 & 0.9319 & 0.9433 & 0.9331 & 0.9432 \\
    1050  & 0.9100 & 0.8935 & 0.8830 & 0.9048 & 0.8830 & 0.9041 & 0.8830 & 0.9052 & 0.8830 \\
    1100  & 0.8800 & 0.8685 & 0.8695 & 0.8683 & 0.8695 & 0.8675 & 0.8695 & 0.8679 & 0.8695 \\
    1150  & 0.8300 & 0.8250 & 0.8409 & 0.8207 & 0.8409 & 0.8198 & 0.8409 & 0.8191 & 0.8409 \\
    1200  & 0.7700 & 0.7555 & 0.7520 & 0.7595 & 0.7520 & 0.7585 & 0.7520 & 0.7560 & 0.7520 \\
    1250  & 0.6850 & 0.6810 & 0.6893 & 0.6809 & 0.6893 & 0.6797 & 0.6893 & 0.6759 & 0.6893 \\
    1300  & 0.5850 & 0.5745 & 0.5797 & 0.5792 & 0.5797 & 0.5783 & 0.5797 & 0.5756 & 0.5797 \\
    1350  & 0.4550 & 0.4365 & 0.4386 & 0.4482 & 0.4388 & 0.4481 & 0.4388 & 0.4527 & 0.4383 \\
    1400  & 0.3100 & 0.2885 & 0.2855 & 0.2913 & 0.2858 & 0.2924 & 0.2858 & 0.3060 & 0.2867 \\
    1450  & 0.1700 & 0.1605 & 0.1505 & 0.1464 & 0.1502 & 0.1489 & 0.1502 & 0.1443 & 0.1485 \\
    1500  & 0.0700 & 0.0743 & 0.0714 & 0.0590 & 0.0714 & 0.0616 & 0.0714 & 0.0537 & 0.0718 \\
    1550  & 0.0350 & 0.0235 & 0.0121 & 0.0216 & 0.0118 & 0.0230 & 0.0117 & 0.0201 & 0.0118 \\
    1600  & 0.0300 & 0.0025 & 0.0009 & 0.0077 & 0.0010 & 0.0082 & 0.0011 & 0.0077 & 0.0010 \\
    \bottomrule
    \end{tabular}
  \label{tab:DigitalPrices}
\end{table}

Finally, digital prices are reported in table \ref{tab:DigitalPrices}.
There are noticeable differences between market prices (although these must be taken with a pretty big pinch of salt
due to the poor liquidity and large bid-ask spreads), the Buchen-Kelly prices and the prices given by the three models.
However, regarding the three relative entropy distributions obtained using the different model priors,
it seems that the effect of the prior density on digital prices is negligible: all three MREDs basically agree with the Buchen-Kelly MED.

\section{Variance Swaps}
\label{s:VarianceSwaps}

In this section we recall the definition of a variance swap and the pricing formula based on replication through a log contract.
(For more details see \cite{CarrMadan1998, DemeterfiDermanKamalZou1999, Kahale2011} and the references therein.)

A variance swap is a forward contract on the annualized realized variance of the underlying asset over a period of time.
More precisely, given observation dates $t_0 < \dots < t_m$, the realized variance is defined by
\begin{equation*}
\sigma_\mathrm{real}^2 := \frac{252}{m} \sum_{i = 1}^m \left[\ln\left(\frac{S(t_i)}{S(t_{i-1})}\right)\right]^2,
\end{equation*}
where $S(t)$ denotes the spot price of the underlying asset at time $t$.
The number $252$ above is the annualization factor and reflects the typical number of business days in a year.
The payoff of a variance swap is given by
\begin{equation*}
N\cdot(\sigma_\mathrm{real}^2 - K_\mathrm{var}),
\end{equation*}
where $K_\mathrm{var}$ is the strike price for variance and $N$ is the notional amount of the swap.

Assume that $\{S(t)\}_{t\ge 0}$ follows a stochastic differential equation
\begin{equation}
\label{ItoProcess_S}
\frac{dS(t)}{S(t)} = \mu(t) dt + \sigma(t) dB(t),
\end{equation}
where $\{B(t)\}_{t\ge 0}$ is a standard Brownian motion and the drift $\{\mu(t)\}_{t\ge 0}$ and the volatility $\{\sigma(t)\}_{t\ge 0}$ are stochastic processes adapted to the natural filtration of $\{B(t)\}_{t\ge 0}$.

Typically, $K_\mathrm{var}$ is such that the theoretical price of the variance swap is null at inception and, in this case,
it is said to be the {\em fair variance swap rate} and denoted by $\sigma^2_\mathrm{fair}$.
The theoretical realized variance over the period $[0, T]$, and so $\sigma^2_\mathrm{fair}$, is given by
\begin{equation*}
\sigma^2_\mathrm{fair} := \frac 1T\ExpectMeas{}{\int_0^T \sigma(t)^2 dt}.
\end{equation*}

We shall now derive a formula for $\sigma^2_\mathrm{fair}$ based on the price of a {\em log contract}, that is, a derivative whose payoff at maturity $T$ is $\ln S(T)$.
Let $x(t) := \ln S(t)$ and apply It\^{o}'s formula to obtain
\begin{equation}
\label{ItoProcess_x}
dx(t) = \left( \mu(t) - \frac{1}{2} \sigma^2(t) \right) dt + \sigma(t) dB(t).
\end{equation}
Subtracting \eqref{ItoProcess_x} from \eqref{ItoProcess_S} gives
\begin{equation}
\label{ItoProcess_SMinusItoProcess_x}
\frac{dS(t)}{S(t)} - dx(t) = \frac{1}{2} \sigma^2(t) dt.
\end{equation}
Integrating from $0$ to $T$ and multiplying by $2/T$ gives
\begin{align*}
\frac1T \int_0^T \sigma^2(t) dt
&=
\frac2T \int_0^T \frac{dS(t)}{S(t)} - \frac2T \int_0^T dx(t) \\
&=
\frac2T \int_0^T \mu(t) dt + \frac2T \int_0^T \sigma(t)dB(t) - \frac2T (x(T) - x(0)).
\end{align*}
Finally, taking expectations yields
\begin{equation}
\label{ExpectationIntegratedItoProcess_SMinusItoProcess_x}
\sigma^2_\mathrm{fair} =
\frac2T \ExpectMeas{}{\int_0^T \mu(t) dt} + \frac2T \ln S(0) - \frac2T \ExpectMeas{}{\ln S(T)},
\end{equation}
since $\ExpectMeas{}{\int_0^T \sigma(t) dB(t)} = 0$.
Notice that $\ExpectMeas{}{\ln S(T)}$ is the price of a log contract.

\subsection{Maximum Entropy and Variance Swaps}

In this section we shall derive a relationship between the fair swap rate of a variance swap and the entropy of the underlying asset density.
This relationship follows from another one relating the entropies of the density $q$ of a random variable $S$ and the density of $x := \ln S$,
which is the subject of the next proposition.

\begin{proposition}
\label{Prop:Ex-Entropy}
The (non-relative) entropy $H(q)$ of a density $q$ of a random variable $S$ on $]0, \infty[$ and the entropy $\tilde{H}(\tilde{q})$ of the density $\tilde{q}$ of $x := \ln(S)$ on $]-\infty, \infty[$ are related by
\begin{equation}
\label{HSHxExpectation}
\tilde{H}(\tilde{q}) - H(q) = \ExpectMeas{}{x}.
\end{equation}
\end{proposition}

\begin{proof} Recall that the densities $q$ and $\tilde{q}$ are related by
\begin{equation*}
q(S) = \frac 1S\ \tilde{q}(\ln S) = \tilde{q}(x)e^{-x}.
\end{equation*}
Hence, the change of measure $dS = e^xdx$, gives
\begin{align*}
H(q) &= \int_0^\infty q(S) \ln q(S) dS
= \int_{-\infty}^\infty \tilde{q}(x) e^{-x} \ln \left( \tilde{q}(x) e^{-x} \right) e^x dx
\\
&= \int_{-\infty}^\infty \tilde{q}(x) \ln \tilde{q}(x) dx - \int_{-\infty}^\infty x \tilde{q}(x) dx
= H(\tilde{q}) - \ExpectMeas{}{x}.
\end{align*}
\end{proof}

\begin{corollary}
\label{Cor:FairVarianceSwapRate-Entropy}
Consider an asset whose price $S(t)$ at time $t$ follows \eqref{ItoProcess_S}.
Let $q$ be the density of $S(T)$, for some $T>0$, and $\tilde{q}$ be the density of $x(T) := \ln S(T)$.
Then the fair variance swap rate of a variance swap maturing at time $T$ is given by
\begin{equation}
\label{HSHxVariance}
\sigma^2_\mathrm{fair} = \frac2T\ExpectMeas{}{\int_0^T \mu(t) dt} + \frac2T \ln S(0) - \frac2T \left( \tilde{H}(\tilde{q}) - H(q) \right).
\end{equation}
\end{corollary}

\begin{proof}
This follows immediately from the last proposition and \eqref{ExpectationIntegratedItoProcess_SMinusItoProcess_x}.
\end{proof}

When the density $q$ of $S(T$) is known, the price of a log contract $\ExpectMeas{}{\ln S(T)}$ can be computed through numerical integration.
Moreover, when $q$ is the MED, that is, in the non-relative entropy case, we show how this price can be computed analytically.
By definition, the expectation is given by
\begin{equation*}
\ExpectMeas{}{\ln S(T)}
= \int_0^\infty \ln(S) q(S) dS
= \sum_{i=0}^n \alpha_i \int_{K_i}^{K_{i+1}} \ln(S) e^{\beta_i S} dS.
\end{equation*}
For $i\notin\{0, n\}$ and $\beta_i \neq 0$ we have
\begin{equation}
\label{EiPrimitiveFunction}
\alpha_i \int_{K_i}^{K_{i+1}} \ln(S) e^{\beta_i S} dS = \frac{\alpha_i}{\beta_i} \left[ e^{\beta_i S} \ln S - \texttt{Ei}(\beta_i S)  \right]_{K_i}^{K_{i+1}},
\end{equation}
where $\texttt{Ei}(s) := -\int_{-s}^\infty \frac{e^{-t}}{t} dt$ is the exponential integral function.
Note that if $\beta_i = 0$, then of course
\begin{equation*}
\int_{K_i}^{K_{i+1}} \ln(S) e^{\beta_i S} dS = \int_{K_i}^{K_{i+1}} \ln(S) dS = \left[ S \ln S - S \right]_{K_i}^{K_{i+1}}.
\end{equation*}
The exponential integral function has a pole at $0$, and therefore we cannot evaluate \eqref{EiPrimitiveFunction} directly at $K_0 = 0$.
From the series representation
\begin{equation*}
\texttt{Ei}(s) = \gamma + \ln |s| + \sum_{k=1}^\infty \frac{s^k}{k \; k!}, \; s\neq 0,
\end{equation*}
it follows that we have in the limit
\begin{equation*}
\lim_{S \to 0} e^{\beta_0 S} \ln S - \texttt{Ei}(\beta_0 S) = -\gamma - \ln |\beta_0|,
\end{equation*}
where $\gamma = 0.5772156649...$ is the Euler-Mascheroni constant.
At the other extreme, at $K_{n+1} = \infty$, we have in the limit
\begin{equation*}
\lim_{S \to \infty} e^{\beta_n S} \ln S - \texttt{Ei}(\beta_n S) = 0,
\end{equation*}
since $\beta_n < 0$.
Putting this together gives a closed formula for \eqref{HSHxExpectation}.

\subsection{Numerical Examples}

In the first example, the market is given as in subsection \ref{ss:BlackScholesMarket} by a Black-Scholes model with volatility $\sigma = 0.25$,
with the same sets of $1, 3$ and $5$ strikes.
Table \ref{tab:VarianceSwap1} shows three quantities obtained from (non-relative) Buchen-Kelly MEDs fitted to the forward and call prices at these strikes:
the fair variance swap rate $\sigma^2_\mathrm{fair}$, its square-root for comparison with implied volatilities, and the entropy.
The average volatility $\sigma_\mathrm{fair}$ and the entropy can be seen as two different measures of the dispersion of $S(T)$.
As the number of strikes increases, $\sigma_\mathrm{fair}$ and the entropy both decrease, with $\sigma_\mathrm{fair}$ converging towards the Black-Scholes volatility $\sigma = 0.25$.

\begin{table}[htbp]
  \centering
  \caption{Fair Variance Swap Rate and Entropy}
    \begin{tabular}{r|rrr}
    \addlinespace
    \toprule
          MED                         & 1 Strike & 3 Strikes & 5 Strikes \\
    \midrule
          $\sigma_\mathrm{fair}$      & 0.3130 & 0.2545 & 0.2506  \\
          $\sigma^2_\mathrm{fair}$    & 0.0980 & 0.0647 & 0.0628  \\
          Entropy                     & 4.6801 & 4.6165 & 4.6077  \\
    \bottomrule
    \end{tabular}
  \label{tab:VarianceSwap1}
\end{table}


In the second example, we use the same reference market as above, but now we include a prior Black-Scholes density and calculate the MREDs matching
the forward and call prices at $1, 3$ and $5$ strikes. The prior density is characterized by its volatility $\sigma_p$.
Table \ref{tab:VarianceSwap2} shows that increasing $\sigma_p$ has the effect of increasing the fair variance swap rate of the MRED.
However, we see that as we add more constraints, this effect is diminished.
In the case of $5$ strikes, it is barely noticeable.
Note that in the case $\sigma_p = 0.25$ where the prior density already matches the given constraints, the MRED is equal to the prior,
and we recover the volatility of the Black-Scholes process as the square-root of the fair variance swap rate.

\begin{table}[htbp]
  \centering
  \caption{Black-Scholes Prior and Fair Variance Swap Rate}
    \begin{tabular}{r|rr|rr|rr}
    \addlinespace
    \toprule
          prior & \multicolumn{2}{|c|}{MRED 1 Strike} & \multicolumn{2}{|c|}{MRED 3 Strikes} & \multicolumn{2}{|c}{MRED 5 Strikes} \\
          BS $\sigma_p$ & $\sigma_\mathrm{fair}$ & $\sigma^2_\mathrm{fair}$ & $\sigma_\mathrm{fair}$ & $\sigma^2_\mathrm{fair}$ & $\sigma_\mathrm{fair}$ & $\sigma^2_\mathrm{fair}$ \\
    \midrule
          0.20  & 0.2427 & 0.0589 & 0.2476 & 0.0613 & 0.2497 & 0.0624 \\
          0.25  & 0.2500 & 0.0625 & 0.2500 & 0.0625 & 0.2500 & 0.0625 \\
          0.30  & 0.2559 & 0.0655 & 0.2514 & 0.0632 & 0.2502 & 0.0626 \\
          0.35  & 0.2608 & 0.0680 & 0.2523 & 0.0637 & 0.2503 & 0.0626 \\
          0.40  & 0.2650 & 0.0702 & 0.2529 & 0.0640 & 0.2503 & 0.0627 \\
          0.45  & 0.2688 & 0.0723 & 0.2533 & 0.0642 & 0.2504 & 0.0627 \\
          0.50  & 0.2723 & 0.0741 & 0.2536 & 0.0643 & 0.2504 & 0.0627 \\
    \bottomrule
    \end{tabular}
  \label{tab:VarianceSwap2}
\end{table}


In the third example, summarized in Table \ref{tab:VarianceSwap3}, we proceed as in the second one, but now with a Heston density as the prior.
The Heston parameters are the same as in subsection \ref{ss:BlackScholesMarket}, i.e. $\kappa = 1, \theta = 0.04, \rho = -0.3, v_0 = 0.04$,
but now we vary the volatility $\sigma$ of the variance and measure its impact on the fair variance swap rate of the MRED.
In the case of $1$ strike, we see clearly that this impact is very strong.
However, we notice again that increasing the number of strikes quickly diminishes the strength of the impact.

\begin{table}[htbp]
  \centering
  \caption{Heston Prior and Fair Variance Swap Rate}
    \begin{tabular}{r|rr|rr|rr}
    \addlinespace
    \toprule
          prior & \multicolumn{2}{|c|}{MRED 1 Strike} & \multicolumn{2}{|c|}{MRED 3 Strikes} & \multicolumn{2}{|c}{MRED 5 Strikes} \\
          Heston $\sigma$ & $\sigma_\mathrm{fair}$ & $\sigma^2_\mathrm{fair}$ & $\sigma_\mathrm{fair}$ & $\sigma^2_\mathrm{fair}$ & $\sigma_\mathrm{fair}$ & $\sigma^2_\mathrm{fair}$ \\
    \midrule
          0.10  & 0.2448 & 0.0599 & 0.2485 & 0.0618 & 0.2499 & 0.0624 \\
          0.20  & 0.2506 & 0.0628 & 0.2500 & 0.0625 & 0.2503 & 0.0627 \\
          0.30  & 0.2600 & 0.0676 & 0.2520 & 0.0635 & 0.2506 & 0.0628 \\
          0.40  & 0.2890 & 0.0835 & 0.2535 & 0.0643 & 0.2507 & 0.0629 \\
          0.50  & 0.3237 & 0.1048 & 0.2544 & 0.0647 & 0.2507 & 0.0629 \\
          0.60  & 0.3464 & 0.1200 & 0.2555 & 0.0653 & 0.2508 & 0.0629 \\
          0.70  & 0.3711 & 0.1377 & 0.2565 & 0.0658 & 0.2508 & 0.0629 \\
    \bottomrule
    \end{tabular}
  \label{tab:VarianceSwap3}
\end{table}

\section{Conclusion}
\label{s:Conclusion}

In this article we generalise the algorithm presented in \cite{NeriSchneider2011} to the relative entropy case.
The algorithm allows for efficient computation of a risk-neutral probability density that exactly gives European
call option prices quoted in the market, while staying as close as possible to a given prior density under the criterion of relative entropy.

It is not necessary to have an analytic expression for the prior density in question. In practice, several popular equity and FX models work through
their characteristic functions and numerical Fourier inversion techniques. We pick two of these as examples, namely the Heston and the Sch\"{o}bel-Zhu
stochastic volatility models, and show how they nevertheless can be used to provide the prior density and incorporated into our algorithm.
In other cases, analytic expressions for the density are available, such as for the Black-Scholes model and the Variance Gamma model%
\footnote{
For example, in C++ this is easily implemented using the functions boost::math::tgamma and boost::math::cyl\_bessel\_k in Boost \cite{Boost2011}.}%
, and we also incorporate these into our analysis.

As an application, we study the impact the choice of prior density has.
In a first, purely hypothetical scenario, we assume that only the prices of a few options are quoted.
We observe that using a prior density does indeed lead to significantly different option prices when
compared to pricing with a pure log-normal density or a piece-wise exponential Buchen-Kelly density.

In a second scenario we use option price data for S\&P$500$ index options for a fixed maturity traded on the CBOE.
We calibrate three different models to these data and observe that, although the models generate noticeably different digital option prices,
the prices obtained when using minimum relative entropy densities, with these models for the prior densities, agree almost perfectly.
Furthermore, these prices are essentially the same as those given by the (non-relative) Buchen-Kelly density itself.
In other words, in a sufficiently liquid market the effect of the prior density seems to vanish almost completely.

We also study variance swaps and establish a formula that relates their fair swap rate to entropy.
In the case of MEDs, we give an explicit formula for the fair swap rate.
In the case of MREDs, we study the impact the prior density has on the fair swap rate and see that, again,
while it has a substantial effect when constraints exist at only a very small number of strikes,
this effect diminishes rapidly as more constraints are added.

\bibliographystyle{plain}
\bibliography{../bibtex/articles,../bibtex/books,../bibtex/websites}

\end{document}